\documentclass{article}
\usepackage{amsmath,amsthm,amssymb}
\newtheorem{Theorem}{Theorem}[section]
\newtheorem{lem}[Theorem]{Lemma}
\newtheorem{Remark}[Theorem]{Remark}
\newtheorem{Definition}[Theorem]{Definition}
\newtheorem{Corollary}[Theorem]{Corollary}
\newtheorem{Proposition}[Theorem]{Proposition}
\newtheorem{Example}[Theorem]{Example}
\numberwithin{equation}{section}

\begin{document}

\title{Constacyclic Codes over $F_p+vF_p$ \footnote{
 E-Mail addresses: zgh09@yahoo.com.cn (G. Zhang);
b\_c\_chen@yahoo.com.cn (B.  Chen).}}

\author{Guanghui Zhang$^1$, Bocong Chen$^2$}

\date{\small 1. Department of Mathematics,
Luoyang Normal University,\\
Luoyang, Henan, 471022, China\\
2. School of Mathematics and Statistics,
Central China Normal University,\\
Wuhan, Hubei, 430079, China}

\maketitle

\begin{abstract}
In this paper, we study constacyclic codes over $F_p+vF_p$, where $p$ is an odd prime and $v^2=v$.
The polynomial generators of all  constacyclic codes over $F_p+vF_p$
are characterized and their dual codes are also determined.


\end{abstract}
\section{Introduction}
Since the  discovery that many good non-linear codes
over finite fields are actually closely related to linear codes over $\mathbb{Z}_4$
via the Gray map (see \cite{HK}), codes over finite rings have received a great deal of
attention (e.g. see \cite{AS}-\cite{GG}, \cite{ZK}).

In these studies,  most of them are concentrated on the case that the ground
rings associated with codes are finite chain rings. However, it turns out that optimal codes over non-chain rings exist.
In  \cite{YK}, Yildiz and Karadeniz   considered  linear codes over the
ring $R_1=F_2+uF_2+vF_2+uvF_2$ with $u^2=v^2=0$ and $uv=vu;$
some good binary codes  were obtained as the images of cyclic codes over $R_1$ under two Gray maps.
In \cite{ZWS}, Zhu, Wang and Shi  studied the structure and properties of cyclic codes  over $F_2+vF_2$, where $v^2=v$;
the authors showed that  cyclic codes over the ring  are principally generated.
In the  subsequent paper \cite{ZW},  Zhu and Wang investigated a class of constacyclic
codes over $F_p+vF_p$ with $p$ being an odd prime and $v^2=v$.
It was proved that the image of a $(1-2v)$-constacyclic code of length $n$ over $F_p+vF_p$
under the Gray map is a distance-invariant cyclic code of length $2n$ over $F_p$
and $(1-2v)$-constacyclic codes over the ring are principally generated.
These rings in the mentioned papers are not finite chain rings.

In this paper, we  generalize the results from   \cite{ZW} to all constacyclic codes over
$R=F_p+vF_p$, where $v^2=v$ and $p$ is an odd prime.
We characterize the polynomial generators of all  constacyclic codes over $F_p+vF_p$,
and
show that constacyclic codes over $R$ of arbitrary length
are principally generated.
The dual codes of the constacyclic codes over $R$ are also discussed.
The rest sections of this paper are organized as follows. In Section 2,
the  preliminary concepts and results are provided. In Section 3,
the polynomial generators of all  constacyclic codes over $F_p+vF_p$
are characterized and their dual codes are also determined.

\section{Preliminaries}
Let $F_p$ be the finite field of order $p$  and $F_p^*$
the multiplicative group of $F_p$,  where $p$ is an odd prime.
It is known that $F_p[x]/\langle x^n-\lambda\rangle$ is a principal ideal ring for
any element $\lambda$ in $F_p^*$.
If $p(x)+\langle x^n-\lambda\rangle\in F_p[x]/\langle x^n-\lambda\rangle$,
then the ideal generated by $p(x)+\langle x^n-\lambda\rangle$, denoted by $\langle p(x)\rangle$,
is the smallest ideal in $F_p[x]/\langle x^n-\lambda\rangle$ containing $p(x)+\langle x^n-\lambda\rangle$.
In addition, we  adopt the notation $[g(x)]$ to denote
the ideal in $F_p[x]/\langle x^n-\lambda\rangle$  generated by $g(x)+\langle x^n-\lambda\rangle$
with $g(x)$ being a monic divisor of $x^n-\lambda$; in that case, $g(x)$
is called a {\it generator polynomial}.

Throughout this paper,  $R$ denotes the commutative ring $F_p+vF_p=\{a+vb|a,b\in F_p\}$ with $v^2=v$.
It turns out that $R$ is a principal ideal ring and has only two non-trivial ideals, namely
$\langle v\rangle=\{va |a\in F_p\}$ and  $\langle 1-v\rangle=\{(1-v)b |b\in F_p\}.$
One can easily check that $\langle v\rangle$ and $\langle 1-v\rangle$  are
maximal ideals in $R$, hence $R$ is not a chain ring.
Let $R^n$ be the $R$-module of $n$-tuples over $R$. A linear code $C$ of length $n$ over $R$ is an
$R$-submodule of $R^n$.
For any linear code $C$ of length $n$ over $R$, the dual $C^\perp$
is defined as $C^\perp=\{u\in R^n\,|\,u\cdot w=0, \mbox{for any $w\in C$}\}$,
where $u\cdot w$ denotes the standard Euclidean inner product of $u$ and $w$ in $R^n$.
Note that $C^\perp$ is linear whether or not $C$ is linear.
The Gray map $\phi$ from $R$ to $F_p\oplus F_p$  given by $\phi(c)=(a,a+b)$,  is a ring isomorphism,
which means that $R$ is isomorphic to the ring $F_p\oplus F_p$. Therefore $R$ is a finite Frobenius ring.
If $C$ is linear, then $|C||C^\perp|=|R|^n$ (See \cite{WJ}).

Let  $\theta$ be a unit in $R$. A linear code $C$ of length $n$ over $R$
is called $\theta$-constacyclic
if for every $(c_0, c_1, \cdots, c_{n-1})\in C$, we have $(\theta c_{n-1}, c_0, c_1, \cdots, c_{n-2})\in C$.
It is well known that a $\theta$-constacyclic code of length $n$ over $R$ can be identified with
an ideal in the quotient ring $R[x]/\langle x^n-\theta\rangle$ via the $R$-module isomorphism as follows:
$$R^n\longrightarrow R[x]/\langle x^n-\theta\rangle$$
$$(c_0, c_1, \cdots, c_{n-1})\longmapsto c_0+c_1x+ \cdots +c_{n-1}x^{n-1}\,\,({\rm mod}\langle x^n-\theta\rangle).$$
If $\theta=1$, $\theta$-constacyclic codes are just cyclic codes and while $\theta=-1$,
$\theta$-constacyclic codes are known as negacyclic codes.

Let $A,B$ be codes over $R$. We denote $A\oplus B = \{a+b|a \in A, b \in B\}$.
Note that any element $c$ of $R^n$ can be expressed as $c=a+vb=v(a+b)+(1-v)a$, where $a,b \in F_p^n$.
Let $C$ be a linear code of length $n$ over $R$. Define
$$
C_{v} = \{b\in F_p^n|va+(1-v)b\in C, \text{for some} \, \,a\in F_p^n\},
$$
and
$$
C_{1-v}= \{a\in F_p^n|va+(1-v)b\in C, \text{for some} \, \,b\in F_p^n\}.
$$
Obviously, $C_{v}$ and $C_{1-v}$ are linear codes over $F_p$. By the definition of $C_{v}$ and $C_{1-v}$,
we have that $C$ can be uniquely expressed as $C=vC_{1-v}\oplus (1-v)C_v$.
It can be routine to check that for any elements $a\in C_{1-v}$ and $b\in C_{v}$,
we get $va+(1-v)b\in C$; in particular, $|C|=|C_{v}||C_{1-v}|$.

\section{Constacyclic codes over $R=F_p+vF_p$}
In this section, we let $R_n=R[x]/\langle x^n-\theta \rangle$ with
$\theta=\lambda+v\mu$ being a unit in $R$, where $\lambda$ and $\mu$
are elements in $F_p$.
As usual, we identify $R_n$ with the set of all polynomials over $R$ of degree less than $n$.
Let $f_1(x), f_2(x),\cdots,f_s(x)\in R_n$.
The ideal generated by $f_1(x), f_2(x),\cdots,f_s(x)$ will be denoted by
$$\langle f_1(x), f_2(x),\cdots,f_s(x)\rangle.$$

The following lemma characterizes the units in $R$.

\begin{lem}
Let $\theta=\lambda+v\mu$ be an element in $R$, where $\lambda$ and $\mu$
are elements in $F_p$. Then $\theta=\lambda+v\mu$ is a unit of $R$
if and only if  $\lambda \neq 0$ and $\lambda+\mu\neq 0$.
\end{lem}
\begin{proof}
($\Longrightarrow$) Suppose that $\theta=\lambda+v\mu$ is a unit of $R$.
Then there exists  elements $a,b\in F_p$ and  $\theta'=a+vb\in R$ such that $\theta\theta'=1$,
that is, $(\lambda+v\mu)(a+vb)=\lambda a+v(\lambda b+\mu a+ \mu b)=1$.
So we have that $\lambda a=1$ and $(\lambda+ \mu)b+\mu a=0$,
which implies that $\lambda \neq 0$ and $\lambda+\mu\neq 0$.

($\Longleftarrow$) Let $\theta=\lambda+v\mu \in R$, where $\lambda \neq 0$ and $\lambda+\mu\neq 0$.
Setting $\theta'=\lambda^{-1}+v[-(\lambda+\mu)^{-1}\mu\lambda^{-1}]$. Then
\begin{eqnarray*}
\theta\theta' & = & (\lambda+v\mu)\big[\lambda^{-1}+v\big(-(\lambda+\mu)^{-1}\mu\lambda^{-1}\big)\big]\\
              & = & 1+v[\mu\lambda^{-1}-\mu(\lambda+\mu)^{-1}-\mu(\lambda+\mu)^{-1}\cdot\mu\lambda^{-1}]\\
              & = & 1+v[\mu\lambda^{-1}-\mu(\lambda+\mu)^{-1}(1+\mu\lambda^{-1})]\\
              & = & 1+v[\mu\lambda^{-1}-\mu(\lambda+\mu)^{-1}(\lambda\lambda^{-1}+\mu\lambda^{-1})]\\
              & = & 1+v[\mu\lambda^{-1}-\mu(\lambda+\mu)^{-1}(\lambda+\mu)\lambda^{-1})]\\
              & = & 1.
\end{eqnarray*}
Hence $\theta=\lambda+v\mu$ is a unit of $R$.
\end{proof}

\begin{Theorem}\label{theorem-1}
Let $C=vC_{1-v}\oplus(1-v)C_v$ be a linear code of length $n$ over $R$.
Then $C$ is a $\theta$-constacyclic code of length $n$ over $R$ if and only if
$C_v$ and $C_{1-v}$ are $\lambda$-constacyclic and $(\lambda+\mu)$-constacyclic codes of length $n$ over $F_p$, respectively.
\end{Theorem}
\begin{proof}
($\Longrightarrow$) Let  $(r_0, r_1,\cdots,r_{n-1})$ be an arbitrary element in  $C_{1-v}$,
and let  $(q_0, q_1,\cdots,q_{n-1})$ be an arbitrary element in $ C_v$.
We assume that $c_i=vr_i+(1-v)q_i, i=0,1,\cdots,n-1;$ hence we get that
$(c_0, c_1,\cdots,c_{n-1})\in C$. Since $C$ is a $\theta$-constacyclic code of length $n$ over $R$,
then $(\theta c_{n-1},c_0,\cdots,c_{n-2})\in C$. Note that
\begin{eqnarray*}
\theta c_{n-1} & = & (\lambda+v\mu)[v r_{n-1}+(1-v)q_{n-1}]\\
               & = & v[(\lambda+\mu)r_{n-1}]+(1-v)(\lambda q_{n-1}),
\end{eqnarray*}
then
\begin{eqnarray*}
(\theta c_{n-1},c_0,\cdots,c_{n-2})& = & v((\lambda+\mu)r_{n-1}, r_0,\cdots,r_{n-2})\\
& + & (1-v)(\lambda q_{n-1}, q_0,\cdots,q_{n-2})\in C,
\end{eqnarray*}
hence $((\lambda+\mu)r_{n-1}, r_0,\cdots,r_{n-2})\in C_{1-v}$ and $(\lambda q_{n-1}, q_0,\cdots,q_{n-2})\in C_v$,
which implies that $C_v$ and $C_{1-v}$ are $\lambda$-constacyclic and $(\lambda+\mu)$-constacyclic codes of length $n$ over $F_p$, respectively.

($\Longleftarrow$) Suppose that $C_v$ and $C_{1-v}$ are $\lambda$-constacyclic and $(\lambda+\mu)$-constacyclic codes of length $n$ over $F_p$, respectively.
Let $(c_0,c_1,\cdots,c_{n-1})\in C$, where $c_i=vr_i+(1-v)q_i, i=0,1,\cdots,n-1$.
It follows that $(r_0,r_1\cdots,r_{n-1})\in C_{1-v}$ and $(q_0,q_1\cdots,q_{n-1})\in C_{v}$.
Note that
\begin{eqnarray*}
(\theta c_{n-1},c_0,\cdots,c_{n-2})&=&v((\lambda+\mu)r_{n-1}, r_0,\cdots,r_{n-2})\\
 & + & (1-v)(\lambda q_{n-1}, q_0,\cdots,q_{n-2})\\
                                   &\in & v C_{1-v}\oplus(1-v)C_v = C,
\end{eqnarray*}
hence $C$ is a $\theta$-constacyclic code of length $n$ over $R$.
\end{proof}
\begin{Theorem}\label{theorem-2}
Let $C=vC_{1-v}\oplus(1-v)C_v$ be a $\theta$-constacyclic code of length $n$ over $R$.
Then $C=\langle vg_{1-v}(x), (1-v)g_v(x)\rangle$, where $g_{1-v}(x)$ and $g_v(x)$ are the generator polynomials of
$C_{1-v}$ and $C_v$, respectively.
\end{Theorem}
\begin{proof}
Since $C_v$ and $C_{1-v}$ are $\lambda$-constacyclic and $(\lambda+\mu)$-constacyclic codes of length $n$ over $F_p$, respectively,
we may assume that the generator polynomials of $C_v$ and $C_{1-v}$ are $g_v(x)$ and $g_{1-v}(x)$, respectively.
Then $vg_{1-v}(x)\in vC_{1-v}\subseteq C$ and $(1-v)g_v(x)\in (1-v)C_v\subseteq C$,
hence $\langle vg_{1-v}(x), (1-v)g_v(x)\rangle\subseteq C$.

Let $f(x)\in C$. Since $C=vC_{1-v}\oplus(1-v)C_v$, then
there are  $s'(x)=g_{1-v}(x)s(x)\in C_{1-v}$ and $u'(x)=u(x)g_v(x)\in C_v$ such that $f(x)=vs'(x)+(1-v)u'(x)=vg_{1-v}(x)s(x)+(1-v)g_v(x)u(x)$,
where $s(x), u(x)\in F_p[x]\subseteq R[x]$. Hence $f(x)\in \langle vg_{1-v}(x), (1-v)g_v(x)\rangle$.
Therefore $C\subseteq \langle vg_{1-v}(x), (1-v)g_v(x)\rangle$. This gives that $C=\langle vg_{1-v}(x), (1-v)g_v(x)\rangle$.
\end{proof}
\begin{Proposition}\label{equation}
Let $C=vC_{1-v}\oplus(1-v)C_v$ be a $\theta$-constacyclic code of length $n$ over $R$
and $g_{1-v}(x), g_v(x)$ are the generator polynomials of $C_{1-v}$ and $C_v$, respectively.
Then $|C|=p^{2n-{\rm deg}(g_{1-v}(x))-{\rm deg}(g_v(x))}$.
\end{Proposition}
\begin{proof}
Since $|C|=|C_v||C_{1-v}|$,
then $|C|=p^{2n-{\rm deg}(g_{1-v}(x))-{\rm deg}(g_v(x))}$.
\end{proof}

For the proof of the following theorem we introduce some notations.
Since the ring $R$ has two maximal ideals $\langle v\rangle$ and $\langle 1-v\rangle$
with the same residue field $F_p$, thus we have two canonical projections defined as follows:
$$\sigma: R=F_p+vF_p \longrightarrow F_p$$
$$va+(1-v)b\longmapsto a;$$
and
$$\tau: R=F_p+vF_p \longrightarrow F_p$$
$$va+(1-v)b\longmapsto b.$$
Denote by $r^{\sigma}$ and $r^{\tau}$ the images of an element $r\in R$ under these two projections, respectively.
These two projections can be extended naturally from $R^n$ to $F_p^n$ and from $R[x]$ to $F_p[x]$.

Let $f(x)=a_0+a_1x+a_2x^2+\cdots+a_{n-1}x^{n-1}$, where $a_i\in R$ for $0\leq i\leq n-1$, and we denote
$$f(x)^{\sigma}=a^{\sigma}_0+a^{\sigma}_1x+\cdots+a^{\sigma}_{n-1}x^{n-1};
f(x)^{\tau}=a^{\tau}_0+a^{\tau}_1x+\cdots+a^{\tau}_{n-1}x^{n-1}.$$
Hence $f(x)$ has a unique expression as $f(x)=vf(x)^{\sigma}+(1-v)f(x)^{\tau}$.

For a code $C$ of length $n$ over $R$, $a\in R$.
The submodule quotient is a linear code of length $n$ over $R$, defined as follows:
$$(C:a)=\{r\in R^n|ar\in C\}.$$

\begin{Theorem}\label{theorem-3}
Let $C=vC_{1-v}\oplus(1-v)C_v$ be a $\theta$-constacyclic code of length $n$ over $R$.
If $C=\langle vh_1(x), (1-v)h_2(x)\rangle$, where $h_1(x)$ and $ h_2(x)\in F_p[x]$ are monic with
$h_1(x)\,|\,(x^n-(\lambda+\mu))$ and $ h_2(x)\,|\,(x^n-\lambda)$,
then $C_{1-v}=[h_1(x)]$ and $C_v=[h_2(x)]$,
that is, $h_1(x)$ and $ h_2(x)$ are the generator polynomials of constacyclic codes $C_{1-v}$ and $C_v$, respectively.
\end{Theorem}
\begin{proof}
We shall prove the theorem by carrying out the following steps.

Step 1. If $C=vC_{1-v}\oplus(1-v)C_v$, then $(C:v)^{\sigma}=C_{1-v}$ and $(C:(1-v))^{\tau}=C_v$.

Let $a\in (C:v)$, then $va\in C$. Setting $a=va^{\sigma}+(1-v)b$, where $b\in F^n_p$.
Hence $va^{\sigma}=v[va^{\sigma}+(1-v)b]=va\in C$. Therefore $a^{\sigma}\in C_{1-v}$,
which implies that $(C:v)^{\sigma}\subseteq C_{1-v}$.
Let $y\in C_{1-v}$, then there exists $z\in F^n_p$ such that $vy+(1-v)z\in C$.
Note that $vy=v[vy+(1-v)z]\in vC\subseteq C$ and $y=vy+(1-v)y$, so $y\in (C:v)$
and $y=y^{\sigma}$, then $y\in (C:v)^{\sigma}$.
Hence $C_{1-v}\subseteq (C:v)^{\sigma}$. Therefore $(C:v)^{\sigma}=C_{1-v}$.

Let $c\in (C:(1-v))$, then $(1-v)c\in C$. Setting $c=va'+(1-v)c^{\tau}$, where $a'\in F^n_p$.
Hence $(1-v)c^{\tau}=(1-v)[va'+(1-v)c^{\tau}]=(1-v)c\in C$. Therefore $c^{\tau}\in C_v$,
which implies that $(C:(1-v))^{\tau}\subseteq C_v$.
Let $y\in C_v$, then there exists $z\in F^n_p$ such that $vz+(1-v)y\in C$.
Note that $(1-v)y=(1-v)[vz+(1-v)y]\in (1-v)C\subseteq C$ and $y=vy+(1-v)y$, so $y\in (C:(1-v))$
and $y=y^{\tau}$. Hence $C_v\subseteq (C:(1-v))^{\tau}$. Therefore $(C:(1-v))^{\tau}=C_v$.

Step 2. If $C=\langle vh_1(x), (1-v)h_2(x)\rangle$,
then $(C:v)^{\sigma}=[h_1(x)]$ and $(C:(1-v))^{\tau}=[h_2(x)]$.

Let $f(x)\in (C:v)$, then $vf(x)\in C$. Since $C=\langle vh_1(x), (1-v)h_2(x)\rangle$,
we have that $vf(x)=vh_1(x)s(x)+(1-v)h_2(x)t(x)$, for some $s(x), t(x)\in R_n$.
Write\small{
$$f(x)=vf(x)^{\sigma}+(1-v)f(x)^{\tau}, s(x)=vs(x)^{\sigma}+(1-v)s(x)^{\tau}~\hbox{and}~ t(x)=vt(x)^{\sigma}+(1-v)t(x)^{\tau},$$}
where $f(x)^{\tau}, s(x)^{\sigma}, s(x)^{\tau}, t(x)^{\sigma}, t(x)^{\tau}\in F_p[x]$. Hence
\begin{eqnarray*}
v[vf(x)^{\sigma}+(1-v)f(x)^{\tau}]& = & v h_1(x)[vs(x)^{\sigma}+(1-v)s(x)^{\tau}]\\
                             & + &(1-v)h_2(x)[vt(x)^{\sigma}+(1-v)t(x)^{\tau}].
\end{eqnarray*}
Thus $vf(x)^{\sigma}=vh_1(x)s(x)^{\sigma}+(1-v)h_2(x)t(x)^{\tau}$, which forces that $f(x)^{\sigma}=h_1(x)s(x)^{\sigma}$.
This shows that $f(x)^{\sigma}\in [ h_1(x)]$. Therefore $(C:v)^{\sigma}\subseteq [ h_1(x)]$.
Conversely, if $f(x)\in [ h_1(x)]$, then $f(x)=h_1(x)u(x)$,
for some $u(x)\in F_p[x]$. Hence $vf(x)=vh_1(x)u(x)\in \langle vh_1(x), (1-v)h_2(x)\rangle=C$,
which shows that $f(x)\in (C:v)$; note that $f(x)=vf(x)+(1-v)f(x)$, so $f(x)=f(x)^{\sigma}$.
Hence $f(x)\in (C:v)^{\sigma}$. We obtain that $[h_1(x)]\subseteq (C:v)^{\sigma}$.
Then we have that $(C:v)^{\sigma}=[h_1(x)]$.

Now we prove the second equality in this step. Let $f(x)\in (C:(1-v))$, then $(1-v)f(x)\in C$.
So we have that $(1-v)f(x)=vh_1(x)d(x)+(1-v)h_2(x)q(x)$, for some $d(x), q(x)\in R_n$.
Write\small{
$$f(x)=vf(x)^{\sigma}+(1-v)f(x)^{\tau}, d(x)=vd(x)^{\sigma}+(1-v)d(x)^{\tau}, q(x)=vq(x)^{\sigma}+(1-v)q(x)^{\tau},$$}
where $f(x)^{\sigma}, d(x)^{\sigma}, d(x)^{\tau}, q(x)^{\sigma}, q(x)^{\tau}\in F_p[x]$.
Thus $(1-v)f(x)^{\tau}=vh_1(x)d(x)^{\sigma}+(1-v)h_2(x)q(x)^{\tau}$, which forces that $f(x)^{\tau}=h_2(x)q(x)^{\tau}$.
This shows that $f(x)^{\tau}\in [h_2(x)]$. Therefore $(C:(1-v))^{\tau}\subseteq [h_2(x)]$.
Conversely, if $f(x)\in [h_2(x)]$, then $f(x)=h_2(x)w(x)$,
for some $w(x)\in F_p[x]$. Hence $(1-v)f(x)=(1-v)h_2(x)w(x)\in \langle vh_1(x), (1-v)h_2(x)\rangle=C$,
which shows that $f(x)\in (C:(1-v))$; note that $f(x)=f(x)^{\tau}$,
hence $f(x)\in (C:(1-v))^{\tau}$. Therefore we obtain that $[h_2(x)]\subseteq (C:(1-v))^{\tau}$,
and thus $(C:(1-v))^{\tau}=[h_2(x)]$.

By the above two steps, we can obtain our desired results.
Specially, $h_1(x)$ and $h_2(x)$ are the generator polynomials of constacyclic codes $C_{1-v}$ and $C_v$, respectively.
\end{proof}
\begin{Definition}
Let $C=vC_{1-v}\oplus(1-v)C_v$ be a $\theta$-constacyclic code of length $n$ over $R$.
We say that the set $S=\{vg_{1}(x), (1-v)g_2(x)\}$ is a generating set in standard form for the $\theta$-constacyclic code
$C=\langle S\rangle$ if both the following two
conditions are satisfied:

(1) For each $i\in \{1,2\}$,  $g_i(x)$ is either monic in $F_p[x]$ or equals to 0;

(2) If $g_1(x) \neq 0$, then $g_1(x)|(x^n-(\lambda+\mu))$;
if $g_2(x) \neq 0$, then $g_2(x)|(x^n-\lambda)$.
\end{Definition}

Now combining Theorem \ref{theorem-2} and \ref{theorem-3}, the following result is obtained.
\begin{Theorem}\label{generatorset}
Any nonzero constacyclic code $C$ over $R$ has a unique generating set in standard form.
\end{Theorem}
\begin{Corollary}\label{principalideal}
Let $C$ be an ideal in $R_n$,
then there exists a unique polynomial $g(x)=vg(x)^{\sigma}+(1-v)g(x)^{\tau}\in C$ such that $C=\langle g(x)\rangle$
with $g(x)^{\sigma}$ and $g(x)^{\tau}$ being monic in $F_p[x]$;
furthermore, $g(x)$ is a divisor of $x^n-\theta$. In particular, $R_n$ is a principal ideal ring.
\end{Corollary}
\begin{proof}
According to  Theorem \ref{generatorset} we have $C=\langle vg_{1-v}(x), (1-v)g_v(x)\rangle$,
where $\{vg_{1-v}(x), (1-v)g_v(x)\}$ is a generating set in standard form for $C$.
Let $g(x)=vg_{1-v}(x)+(1-v)g_v(x)$. Note that
$$vg_{1-v}(x)=v[vg_{1-v}(x)+(1-v)g_v(x)],$$
and
$$(1-v)g_v(x)=(1-v)[vg_{1-v}(x)+(1-v)g_v(x)],$$
which imply that $C=\langle g(x)\rangle$.
Note that there exist polynomials $r_{1-v}(x)$ and $r_v(x)$ in $F_p[x]$ such that
$$
x^n-(\lambda+\mu)=g_{1-v}(x)r_{1-v}(x),
$$
$$
x^n-\lambda=g_{v}(x)r_{v}(x).
$$
Then we get
$$
x^n-\theta=g(x)[vr_{1-v}(x)+(1-v)r_v(x)].
$$

Finally, we prove the uniqueness of such a polynomial.
Suppose that $C=\langle h(x)\rangle$. Write $h(x)=v h(x)^{\sigma}+(1-v)h(x)^{\tau}$,
where $h(x)^{\sigma}$ and $h(x)^{\tau}$ are monic in $F_p[x]$.
In the following we shall prove that $h(x)^{\sigma}=g_{1-v}(x)$
and $h(x)^{\tau}=g_v(x)$.
Since $C=\langle h(x)\rangle$ and  $vh(x)\in C$, so $h(x)\in (C:v)$, that is, $h(x)^{\sigma}\in (C:v)^{\sigma}=C_{1-v}$.
Then $g_{1-v}(x)|h(x)^{\sigma}$; similarly, we have that $g_{v}(x)|h(x)^{\tau}$.
On the other hand, there exists some polynomial $s(x)\in R_n$ such that
$$v g_{1-v}(x)+(1-v)g_v(x)=[v s(x)^{\sigma}+(1-v)s(x)^{\tau}][v h(x)^{\sigma}+(1-v)h(x)^{\tau}],$$
\small{it follows that $s(x)^{\sigma}h(x)^{\sigma}=g_{1-v}(x)$ and $s(x)^{\tau}h(x)^{\tau}=g_v(x)$.
Hence $h(x)^{\sigma}\,|\,g_{1-v}(x)$} and  $h(x)^{\tau}\,|\,g_v(x)$.
Therefore we obtain that $h(x)^{\sigma}=g_{1-v}(x)$
and $h(x)^{\tau}=g_v(x)$, which is the required results.
\end{proof}

\begin{Remark}
We mention that,  the approach used in the above results  also valid when $p=2$.
In particular, Corollary~\ref{principalideal} yields a generalization of a main result in
\cite{ZWS}, which showed that cyclic codes over $F_2+vF_2$ are principally generated.
In the following, we shall see that the hypothesis with $p$ being an odd prime is necessary
on the discussion of $(1-2v)$ or $(-1+2v)$-constacyclic codes over $R$.
\end{Remark}

Now we give the definition of polynomial Gray map over $R_n$.
Let $f(x)\in R_n$ with degree less than $n$, then $f(x)$ can be expressed as
$f(x)=r(x)+vq(x)$, where $r(x), q(x)\in F_p[x]$ and their degrees are both less than
$n$. Let $\theta = \lambda+v \mu\in R^*$. Define the polynomial Gray map as follows:
$$\phi_{\theta}: R_n\longrightarrow F_p[x]/\langle x^{2n}-1\rangle$$
$$f(x)=r(x)+vq(x) \longmapsto \lambda(\lambda+\mu)q(x)+x^n[-\mu r(x)-(\lambda+\mu)q(x)].$$
Obviously the above polynomial Gray map $\phi_{\theta}$ is well-defined.
If  $\mu\neq 0$, then the map $\phi_{\theta}$ is bijective.
\begin{Theorem}\label{grayimage}
Let $C$ be a $\theta$-constacyclic code of length $n$ over $R$ with a generating set
in standard form $\{v g_{1-v}(x),(1-v)g_v(x)\}$.
Then $$\phi_{\theta}(C)\subseteq \langle g_{1-v}(x)g_v(x)\rangle.$$
\end{Theorem}
\begin{proof}
Since $g_{1-v}(x)|(x^n-(\lambda+\mu))$ and $g_v(x)|(x^n-\lambda)$,
there exists $q_1(x), q_2(x)\in F_p[x]$ such that
$$x^n-(\lambda+\mu)=g_{1-v}(x)q_1(x) \,\,\, \text{and}\,\,\, x^n-\lambda=g_v(x)q_2(x).$$
By the proof of Corollary \ref{principalideal}, we have that $C=\langle vg_{1-v}(x)+(1-v)g_v(x)\rangle$.
Let $f(x)$ be any element in $C$. Then $f(x)=[vg_{1-v}(x)+(1-v)g_v(x)]h(x)$, for some $h(x)\in R_n$.
Since $h(x)$ can be written as $h(x)=vh(x)^{\sigma}+(1-v)h(x)^{\tau}$, where $h(x)^{\sigma}, h(x)^{\tau}\in F_p[x]$,
it follows that $f(x)=g_v(x)h(x)^{\tau}+v[g_{1-v}(x)h(x)^{\sigma}-g_v(x)h(x)^{\tau}]$. Then we have that
\begin{eqnarray*}
\phi_{\theta}(f(x)) & = & \lambda(\lambda+\mu)[g_{1-v}(x)h(x)^{\sigma}-g_v(x)h(x)^{\tau}]\\
           & + & x^n[-\mu g_v(x)h(x)^{\tau}-(\lambda+\mu)(g_{1-v}(x)h(x)^{\sigma}-g_v(x)h(x)^{\tau})]\\
           & = & \lambda g_v(x)h(x)^{\tau}(x^n-(\lambda+\mu))-(\lambda+\mu)g_{1-v}(x)h(x)^{\sigma}(x^n-\lambda)\\
           & =& \lambda g_v(x)h(x)^{\tau}g_{1-v}(x)q_1(x)-(\lambda+\mu)g_{1-v}(x)h(x)^{\sigma}g_v(x)q_2(x)\\
           & = & g_{1-v}(x)g_v(x)[\lambda h(x)^{\tau}q_1(x)-(\lambda+\mu)h(x)^{\sigma}q_2(x)].
\end{eqnarray*}
Hence $\phi_{\theta}(C)\subseteq \langle g_{1-v}(x)g_v(x)\rangle$.
\end{proof}
\begin{Corollary}\label{grayimage-2}
Let $\vartheta=1-2v \,\,\,\text{or}\,\,\, -1+2v$ and let
$C$ be a $\vartheta$-constacyclic code of length $n$ over $R$ with a generating set
in standard form $\{v g_{1-v}(x),(1-v)g_v(x)\}$.
Then $\phi_{\vartheta}(C)=[g_{1-v}(x)g_v(x)]$.
\end{Corollary}
\begin{proof}
Note that $g_{1-v}(x)|(x^n-(\lambda+\mu))$ and $g_v(x)|(x^n-\lambda)$,
where $\lambda+v\mu=1-2v \,\,\,\text{or}\,\,\, -1+2v$,
then $(x^n-(\lambda+\mu))(x^n-\lambda)=x^{2n}-1$,
hence $g_{1-v}(x)g_v(x)|(x^{2n}-1)$, which shows that $g_{1-v}(x)g_v(x)$ is the generator polynomial
for cyclic code $\langle g_{1-v}(x)g_v(x)\rangle$, that is, $\langle g_{1-v}(x)g_v(x)\rangle=[g_{1-v}(x)g_v(x)]$.

By Theorem \ref{grayimage}, we have that $\phi_{\vartheta}(C)\subseteq [g_{1-v}(x)g_v(x)]$.
On the other hand,
$$|\phi_{\vartheta}(C)|=|C|=p^{2n-{\rm deg}(g_{1-v}(x))-{\rm deg}(g_v(x))}$$
and
$$|[ g_{1-v}(x)g_v(x)]|= p^{2n-{\rm deg}(g_{1-v}(x))-{\rm deg}(g_v(x))}.$$
Hence, $\phi_{\vartheta}(C)=[ g_{1-v}(x)g_v(x)]$.
\end{proof}

The above Corollary \ref{grayimage-2} shows that the Gray image of a $\vartheta$-constacyclic code
over $R$ under the Gray map $\phi_{\vartheta}$ is a cyclic code of length $2n$ over $F_p$.
In order to study the converse part, we now give the corresponding Gray map on $R^n$.
Let $\theta=\lambda+v\mu$,
$$\phi_{\theta}: R^n\longrightarrow F_p^{2n}$$
$$(c_0,c_1,\cdots,c_{n-1})\longmapsto \big(\lambda(\lambda+\mu)q_0,\lambda(\lambda+\mu)q_1,\cdots,\lambda(\lambda+\mu)q_{n-1},$$
$$-\mu r_0-(\lambda+\mu)q_0,-\mu r_1-(\lambda+\mu)q_1,\cdots,-\mu r_{n-1}-(\lambda+\mu)q_{n-1}\big).$$
where $c_i=r_i+vq_i, 0\leq i\leq n-1$.
\begin{lem}\label{lem-2}
Let $\vartheta=1-2v \,\,\,\text{or}\,\,\, -1+2v$ and $\phi_{\vartheta}$ be the Gray map of $R^n$ into $F_p^{2n}$.
Let $\alpha$ be the $\vartheta$-constacyclic shift of $R^n$ and $\beta$ the cyclic shift of $F_p^{2n}$.
Then $\phi_{\vartheta}\alpha=\beta\phi_{\vartheta}$.
\end{lem}
\begin{proof}
We first consider the case when $\vartheta=-1+2v$. Let $c=(c_0,c_1,\cdots,c_{n-1})\in R^n$,
where $c_i=r_i+vq_i, 0\leq i\leq n-1$. By the definition of $\phi_{\vartheta}$,
we have that
$$\phi_{\vartheta}(c)=(-q_0,-q_1,\cdots,-q_{n-1},-2r_0-q_0,-2r_1-q_1,\cdots,-2r_{n-1}-q_{n-1}).$$
Hence
$$\beta(\phi_{\vartheta}(c))=(-2r_{n-1}-q_{n-1},-q_0,\cdots,-q_{n-1},-2r_0-q_0,\cdots,-2r_{n-2}-q_{n-2}).$$
On the other hand,
\begin{eqnarray*}
\alpha(c) & = & (\vartheta c_{n-1}, c_0, \cdots,c_{n-2})\\
          & = & (-r_{n-1}+v(2r_{n-1}+q_{n-1}), r_0+v q_0, \cdots, r_{n-2}+v q_{n-2}).
\end{eqnarray*}
So we obtain
$$\phi_{\vartheta}(\alpha(c))=(-2r_{n-1}-q_{n-1},-q_0,\cdots,-q_{n-1},-2r_0-q_0,\cdots,-2r_{n-2}-q_{n-2}).$$
Therefore, $\phi_{\vartheta}\alpha=\beta\phi_{\vartheta}$.

The case when $\vartheta=1-2v$ is [\cite{ZW}, Proposition 3.3].
\end{proof}
\begin{Theorem}
Let $\vartheta=1-2v \,\,\,\text{or}\,\,\, -1+2v$.
A linear code $C$ of length $n$ over $R$ is a $\vartheta$-constacyclic code if and only if
$\phi_{\vartheta}(C)$ is a cyclic code of length $2n$ over $F_p$.
\end{Theorem}
\begin{proof}
It is obtained by Lemma \ref{lem-2}.
\end{proof}

The following result is routine to check or see the proof of [\cite{DI}, Proposition 2.4].
\begin{Proposition}\label{thedual}
Let $C$ be a $\theta$-constacyclic code of length $n$ over $R$.
Then the dual code $C^\perp$ for $C$ is a $\theta^{-1}$-constacyclic code of length $n$ over $R$.
\end{Proposition}

Let $g_{1-v}(x)h_{1-v}(x)=x^n-(\lambda+\mu), g_v(x)h_v(x)=x^n-\lambda$.
Let $\widetilde{h}_{1-v}(x)=x^{{\rm deg}(h_{1-v}(x))}h_{1-v}(\frac{1}{x})$ and
$\widetilde{h}_v(x)=x^{{\rm deg}(h_v(x))}h_v(\frac{1}{x})$ be the reciprocal polynomials of $h_{1-v}(x)$ and $h_v(x)$, respectively.
We write $h_{1-v}^*(x)=\frac{1}{h_{1-v}(0)}\widetilde{h}_{1-v}(x)$ and $ h_v^*(x)=\frac{1}{h_v(0)}\widetilde{h}_v(x)$.

\begin{Theorem}\label{dualdecomposition}
Let $C=vC_{1-v}\oplus(1-v)C_v$ be a $\theta$-constacyclic code of length $n$ over $R$.
Then $C^\perp=vC^\perp_{1-v}\oplus(1-v)C^\perp_v$.
\end{Theorem}
\begin{proof}
From Theorem \ref{theorem-1}, $C_{1-v}$ and $C_v$ are constacyclic codes over $F_p$,
then $C_{1-v}^\perp$ and $C^\perp_v$ are also constacyclic codes $F_p$.
Let $g_{1-v}(x)$ and $g_v(x)$ are generator polynomials for $C_{1-v}$ and $C_v$, respectively.
Then $C_{1-v}^\perp=[h_{1-v}^*(x)]$ and  $C_v^\perp=[h_v^*(x)]$.
Thus we have that $|C_{1-v}^\perp|=p^{{\rm deg}(g_{1-v}(x))}$ and $|C_v^\perp|=p^{{\rm deg}(g_v(x))}$.

For any $a\in C^\perp_{1-v}, b\in C_v^\perp$ and $c=vr+(1-v)q\in C$,
where $r\in C_{1-v}, q\in C_v$, we have that
\begin{eqnarray*}
c\cdot(va+(1-v)b) & = & (v r+(1-v)q)\cdot(va+(1-v)b) \\
                  & = & v(r\cdot a)+(1-v)(q\cdot b) \\
                  & = & 0,
\end{eqnarray*}
and hence $vC^\perp_{1-v}\oplus(1-v)C^\perp_v \subseteq C^\perp$.

Furthermore, suppose that $va+(1-v)b=va'+(1-v)b'$, where $a, a' \in C_{1-v}^\perp$ and $ b, b' \in C_v^\perp$,
then $v(a-a')=(1-v)(b'-b)$, so
$$v(a-a')=v[v(a-a')]=v[(1-v)(b-b')]=0.$$
Hence $a=a'$, which forces $b=b'$.
Thus every element $c$ of $vC^\perp_{1-v}\oplus(1-v)C^\perp_v$ has a unique expression as
$c=vr+(1-v)q$, where $r \in C_{1-v}^\perp, q\in C_v^\perp$.
This shows that
\begin{eqnarray*}
|vC^\perp_{1-v}\oplus(1-v)C^\perp_v| & = & |C_{1-v}^\perp||C_v^\perp|\\
                                     & = & p^{{\rm deg}(g_{1-v}(x))+{\rm deg}(g_v(x))}.
\end{eqnarray*}

Finally, by Proposition \ref{equation}, $|C|=p^{2n-{\rm deg}(g_{1-v}(x))-{\rm deg}(g_v(x))}$.
Since $R$ is a Frobenius ring, $|C||C^\perp|=|R|^n$, so
\begin{eqnarray*}
|C^\perp| & = & \frac{|R|^n}{|C|}\\
          & = & \frac{p^{2n}}{p^{2n-{\rm deg}(g_{1-v}(x))-{\rm deg}(g_v(x))}}\\
          & = & p^{{\rm deg}(g_{1-v}(x))+{\rm deg}(g_v(x))}\\
          & = & |vC^\perp_{1-v}\oplus(1-v)C^\perp_v|.
\end{eqnarray*}
Note that $vC^\perp_{1-v}\oplus(1-v)C^\perp_v \subseteq C^\perp$ as above,
we have that $C^\perp =vC^\perp_{1-v}\oplus(1-v)C^\perp_v$, as required.
\end{proof}

According to the above results and their proofs, we can carry over the results regarding constacyclic
codes corresponding to their dual codes.
\begin{Theorem}
With notations as above. Let $C$ be a $\theta$-constacyclic code of length $n$ over $R$ with a generating set
in standard form $\{v g_{1-v}(x),(1-v)g_v(x)\}$. Then

(1) $C^\perp=\langle v h_{1-v}^*(x), (1-v)h_v^*(x)\rangle$ and $|C^\perp|=p^{{\rm deg}(g_{1-v}(x))+{\rm deg}(g_v(x))}$;

(2) $C^\perp=\langle v h_{1-v}^*(x)+(1-v)h_v^*(x)\rangle$;

(3) $\phi_{\theta}(C^\perp)\subseteq\langle h^*_{1-v}(x)h^*_v(x)\rangle$;
\end{Theorem}
\begin{proof}
(1) By Proposition \ref{thedual}, $C^\perp$ is a $\theta^{-1}$-constacyclic code over $R$;
by Theorem \ref{dualdecomposition}, we have that $C^\perp=vC^\perp_{1-v}\oplus(1-v)C^\perp_v$,
where according to Theorem \ref{theorem-1} $C_{1-v}^\perp$ and $C^\perp_v$ are two constacyclic codes
over $F_p$. Since $h_{1-v}^*(x)$ and $h^*_v(x)$ are generator polynomials for $C_{1-v}^\perp$ and $C^\perp_v$, respectively,
we have that $\{vh_{1-v}^*(x), (1-v)h^*_v(x)\}$ is the generating set in standard form for $C^\perp$.
So $C^\perp=\langle v h_{1-v}^*(x), (1-v)h_v^*(x)\rangle$.
In addition, $|C^\perp|=|C_{1-v}^\perp||C^\perp_v|=p^{{\rm deg}(g_{1-v}(x))}\cdot p^{{\rm deg}(g_v(x))}=p^{{\rm deg}(g_{1-v}(x))+{\rm deg}(g_v(x))}$.

(2) Since $\{vh_{1-v}^*(x), (1-v)h^*_v(x)\}$ is the generating set in standard form for $C^\perp$,
according to the proof of Corollary \ref{principalideal} we have that $C^\perp=\langle v h_{1-v}^*(x)+(1-v)h_v^*(x)\rangle$.

(3)Similar to the proof of Theorem \ref{grayimage}.
\end{proof}
\begin{Theorem}
Let $\vartheta=1-2v \,\,\,\text{or}\,\,\, -1+2v$ and let
$C$ be a $\vartheta$-constacyclic code of length $n$ over $R$ with a generating set
in standard form $\{v g_{1-v}(x),(1-v)g_v(x)\}$. Then

(1) $\phi_{\vartheta}(C^\perp)=[h^*_{1-v}(x)h^*_v(x)]$;

(2) $\phi_{\vartheta}(C^\perp)=(\phi_{\vartheta}(C))^\perp$.
\end{Theorem}
\begin{proof}
(1) According to the proof of  Corollary \ref{grayimage-2}, we can obtain the result.

(2) Note that the facts that
$$\phi_{\vartheta}(C)=[g_{1-v}(x)g_v(x)]
~~\hbox{and}~~ \phi_{\vartheta}(C^\perp)=[ h^*_{1-v}(x)h^*_v(x)],$$ we have
\begin{eqnarray*}
\phi_{\vartheta}(C)^\perp & = & [g_{1-v}(x)g_v(x)] ^\perp\\
                       & = & [ h^*_{1-v}(x)h^*_v(x)] \\
                       & = & \phi_{\vartheta}(C^\perp),
\end{eqnarray*}
which is the required result.
\end{proof}

\begin{Example}
In $F_3[x]$,
$$x^{10}+1=(x^2+1)(x^4+x^3-x+1)(x^4-x^3+x+1);$$
$$x^{10}-1=(x-1)(x+1)(x^4+x^3+x^2+x+1)(x^4-x^3+x^2-x+1).$$
Let $C$ be the $(-1+2v)$-constacyclic code of length $10$ over $F_3+vF_3$ with
generating polynomial
\begin{eqnarray*}
g(x) & = & v(x^4+x^3-x+1)+(1-v)(x^4-x^3+x^2-x+1)\\
& = & x^4+(2v-1)x^3+(1-v)x^2-x+1.
\end{eqnarray*}
The Gray image $\phi_{\vartheta}(C)$ is a $[20, 12, 4]$ cyclic code over $F_3$ with generator polynomial
$(x^4+x^3-x+1)(x^4-x^3+x^2-x+1)$.
\end{Example}

\noindent{\bf Acknowledgement}
This work was supported by NSFC, Grant No. 11171370.



\begin{thebibliography}{21}

\bibitem{AS} T. Abualrub, I. Siap, Constacyclic codes over $F_2+uF_2$, J. Franklin Inst. 346(2009)520-529.

\bibitem{BLA} T. Blackford, Negacyclic codes over $\mathbb{Z}_4$ of even length, IEEE Trans. Theory 49(2003)1417-1424.


\bibitem{DI} H. Q. Dinh, Constacyclic codes of length $p^s$ over $F_{p^m}+uF_{p^m}$, J. Algebra 324(2010)940-950.

\bibitem{DIN} H. Q. Dinh, Constacyclic codes of length $2^s$ over Galois extension rings of $F_2+uF_2$, IEEE Trans. Theory
55(2009)1730-1740.

\bibitem{DINH} H. Q. Dinh, S. R. L\'{o}pez-Permouth, Cyclic and negacyclic codes over finite chain rings,
IEEE Trans. Inform. Theory 50(8)(2004)1728-1744.

\bibitem{GG} K. Guenda, T. A. Gulliver, MDS and self-dual codes over rings, Finite Fields Appl. 18(2012)1061-1075.

\bibitem{HK} A. R. Hammons , Jr.,  P. V. Kumar,  A. R. Calderbank,
N. J. A. Sloane,  P. Sol\'{e},  The $\mathbb{Z}_4$ linearity of Kerdock,
Preparata,  Goethals and related codes,  IEEE Trans. Inform. Theory 40(2)(1994)
301-319.

\bibitem{WJ} J. Wood, Duality for modules over finite rings and applications to coding theory, Amer. J. Math. 121(1999) 555-575.

\bibitem{YK} B. Yildiz, S. Karadeniz, Linear codes over $F_2+uF_2+vF_2+uvF_2$, Des. Codes Cryptogr. 54(2010)61-81.

\bibitem{ZK} S, Zhu, X. Kai, Dual and self-dual negacyclic codes of even length over $\mathbb{Z}_{2^a}$,
Discrete Math. 309(2009)2382-2391.

\bibitem{ZW} S. Zhu, L. Wang, A class of constacyclic codes over $F_p+vF_p$ and its Gray image, Discrete Math. 311(2011)2677-2682.

\bibitem{ZWS} S. Zhu, Y. Wang, M. Shi, Some results on cyclic codes over $F_2+vF_2$, IEEE Trans. Inform. Theory 56(4)(2010)1680-1684.

\end{thebibliography}
\end{document}